\documentclass[11pt, a4paper]{article}

\usepackage{amsmath,amssymb,amsthm}
\usepackage{mathtools}
\usepackage{graphicx}
\usepackage{caption, subcaption}
\usepackage{fullpage}
\usepackage{color}
\usepackage[colorlinks=true, citecolor=blue]{hyperref}
\usepackage[usenames,dvipsnames]{xcolor}
\usepackage{tikz}
\usepackage{authblk}

\RequirePackage[capitalize,nameinlink,noabbrev]{cleveref}

\usepackage[T1]{fontenc}
\usepackage{libertinus}

\usepackage{lipsum}

\newtheorem{thm}{Theorem}

\newtheorem{lem}{Lemma}

\numberwithin{equation}{section}

\title{Exact solution of weighted partially directed walks crossing a square}

\author{Nicholas R.\ Beaton\thanks{\href{mailto:nrbeaton@unimelb.edu.au}{nrbeaton@unimelb.edu.au}}\ }
\author{Aleksander L.\ Owczarek\thanks{\href{mailto:owczarek@unimelb.edu.au}{owczarek@unimelb.edu.au}}}
\affil{School of Mathematics and Statistics \\ The University of Melbourne \\ Parkville, VIC, 3010, Australia}

\usepackage[maxbibnames=99,sorting=nyt,sortcites=true,style=numeric-comp]{biblatex}
\renewbibmacro{in:}{}
\DeclareFieldFormat[misc]{title}{\mkbibquote{#1}}

\DeclareFieldFormat[article]{volume}{\mkbibbold{#1}}
\DeclareFieldFormat{url}{\textsc{url}: \href{#1}{#1}}
\DeclareFieldFormat{doi}{\textsc{doi}: \href{http://dx.doi.org/#1}{#1}}
\DeclareFieldFormat{eprint:arxiv}{arXiv:\href{https://arxiv.org/abs/#1}{#1}}
\begin{document}

\maketitle
\begin{abstract}
We consider partially directed walks crossing a $L\times L$ square weighted according to their length by a fugacity $t$. The exact solution of this model is computed in three different ways, depending on whether $t$ is less than, equal to or greater than 1. In all cases a complete expression for the dominant asymptotic behaviour of the partition function is calculated. The model admits a dilute to dense phase transition, where for $0 < t < 1$ the partition function scales exponentially in $L$ whereas for $t>1$ the partition function scales exponentially in $L^2$, and when $t=1$ there is an intermediate scaling which is exponential in $L \log{L}$. 

\end{abstract}

\section{Introduction}\label{sec:intro}

The problem of self-avoiding walks (SAWs) crossing a square~\cite{Whittington1990,Madras1995,BousquetMelou2005,Knuth1976}, or walks or polygons simply contained in a square~\cite{Bradly2022,Guttmann2022,Guttmann2022b} in two dimensions, or inside a cubic box in three dimensions~\cite{Whittington2022}, has attracted attention over an extended period including recently, with various rigorous and numerical (Monte Carlo and series analysis) results being accumulated. These problems provide a simple model of a confined polymer which illustrate a different lens through which to consider single polymer behaviour. When a length fugacity is added to the basic set-up the models can be shown to demonstrate a phase transition between a dilute phase for low fugacity and a dense phase for large fugacity \cite{Whittington1990,Madras1995,BousquetMelou2005}. The scaling of the partition function is fundamentally different in these two regimes with exponential scaling linear in the side of the square (box) in the dilute phase and exponential in the area of the square (volume of the box) in the dense phase. 

For example, let $c_{L,n}$ be the number of $n$-step SAWs on the square lattice which cross an $L\times L$ square from the south-west corner to the north-east corner, and define the partition function
\begin{equation}
C_L(t) = \sum_n c_{L,n}t^n.
\end{equation}
Then it is known rigorously~(e.g.~\cite{Madras1995,Whittington1990}) that the limits
\begin{align}
\lambda_1(t) &= \lim_{L\to\infty} C_L(t)^{1/L} \\
\lambda_2(t) &= \lim_{L\to\infty} C_L(t)^{1/L^2}
\end{align}
exist or are infinite. More precisely, $\lambda_1(t)$ is finite for $0<t\leq\mu^{-1}$ and infinite for $t>\mu^{-1}$, where $\mu$ is the connective constant of the lattice; and $\lambda_2(t)=1$ for $0<t\leq\mu^{-1}$ and is finite and $>1$ for $t>\mu^{-1}$. Moreover $\lambda_1(t) <1$ for $t<\mu^{-1}$ and $\lambda_1(\mu^{-1}) = 1$; otherwise the values of $\lambda_1(t)$ and $\lambda_2(t)$ are not known for $t<\mu^{-1}$ and $t>\mu^{-1}$ respectively. These results generalise to higher dimensions. The precise nature of the `subexponential' behaviour of $C_L(t)$ is not known, however it has been recently shown~\cite{Whittington2022} that
\begin{equation}
C_L(1) = \lambda^{L^2+O(L)}
\end{equation}
with $\lambda=\lambda_2(1)$. A similar result holds for higher dimensions. This was motivated by the conjecture \cite{Guttmann2022,Guttmann2022b} that
\begin{equation}
C_L(1) \sim \lambda^{L^2+b L+c} L^g
\end{equation}
for constants $b$, $c$ and $g$. Note that here and below in the sequel the notation $a_L \sim b_L$ indicates that $\lim_{L\to\infty} \frac{a_L}{b_L} = 1$.

Here we consider a variation of this model, namely \emph{partially directed walks} (PDWs) crossing an $L\times L$ square. These are walks which take steps $(1,0), (0,1)$ and $(0,-1)$ while remaining self-avoiding. This is, of course, a  simpler model than SAWs, but directed and partially directed walks have been shown to display complex critical behaviour for a range of models, from adsorption to collapse (see e.g.~\cite{forgacs1991b-a,owczarek1993b-:a,Zhou:2006tq,Owczarek:2007jy,VANRENSBURG2008623,Owczarek:2009dp,Brak:2009fk,Lam:2009ud,Owczarek:2010bv,Rensburg2016a,Legrand2022a}). Here we compute the exact solution of PDWs crossing a square and provide the full dominant asymptotics of the partition function as a function of the length fugacity $t$.

For PDWs the dilute-dense phase transition occurs at $t=1$. Interestingly, each regime (dilute, dense, and at the critical point) requires a different mathematical approach to elucidate the solution. For small $t<1$ the generating function is found via the kernel method, and the asymptotics of the partition function follow via  saddle point methods. For large $t>1$ a transfer matrix method is required, and is analysed with a Bethe ansatz type solution and the asymptotics follow a subtle analysis of the Bethe roots. The solution at $t=1$ is simply found via a direct combinatorial argument.


\section{Model and central results}\label{sec:model}

Let $\mathcal{P}_{L,n}$ be the set of $n$-step PDWs which cross an $L\times L$ square from the south-west corner to the north-east corner, and let $p_{L,n} = |\mathcal{P}_{L,n}|$. Define the partition function
\begin{equation}
P_L(t) = \sum_n p_{L,n} t^n.
\end{equation}
For a given value of $t>0$, the Boltzmann distribution on $\mathcal{P}_L = \bigcup_n \mathcal{P}_{L,n}$ assigns probability
\begin{equation}
\mathbb{P}_L(t,\omega) = \frac{t^{|\omega|}}{P_L(t)}
\end{equation}
to the PDW $\omega$, where $|\omega|$ is the length of $\omega$. See \cref{fig:samples} for some PDWs in the box of size $L=20$ sampled from the Boltzmann distribution at various values of $t$.

We then define the mean number of steps for walks in the $L\times L$ square to be
\begin{equation}\label{eqn:meanlength}
\langle n\rangle_{L} = \frac{\sum_n n p_{L,n}t^n}{\sum_n p_{L,n}t^n} = \frac{t\frac{d}{dt}P_{L}(t)}{P_L(t)}.
\end{equation}

Our main result is the following.
\begin{figure}[h]
\centering
\begin{subfigure}{0.32\textwidth}
\includegraphics[width=\textwidth]{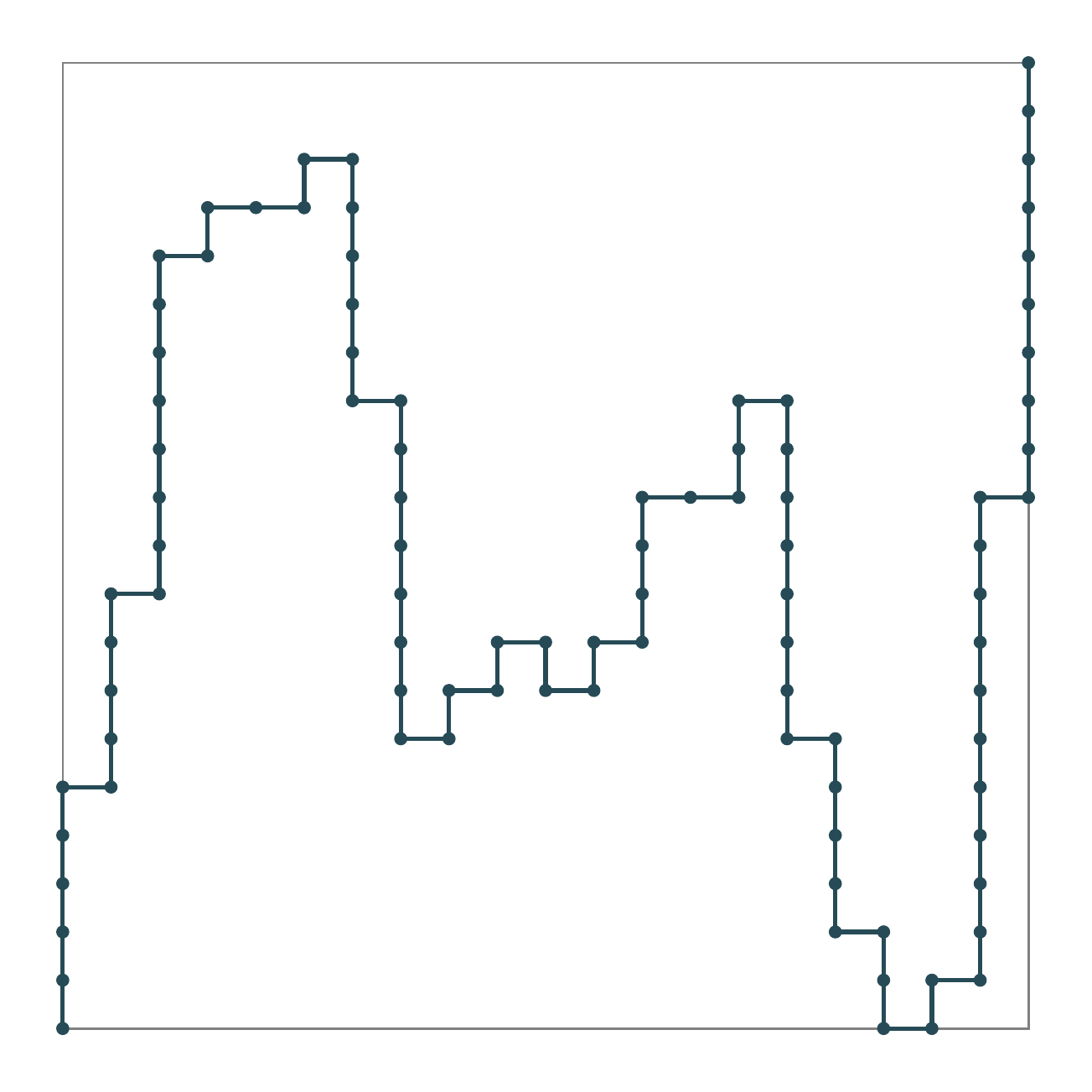}
\caption{}
\end{subfigure}
\hfill
\begin{subfigure}{0.32\textwidth}
\includegraphics[width=\textwidth]{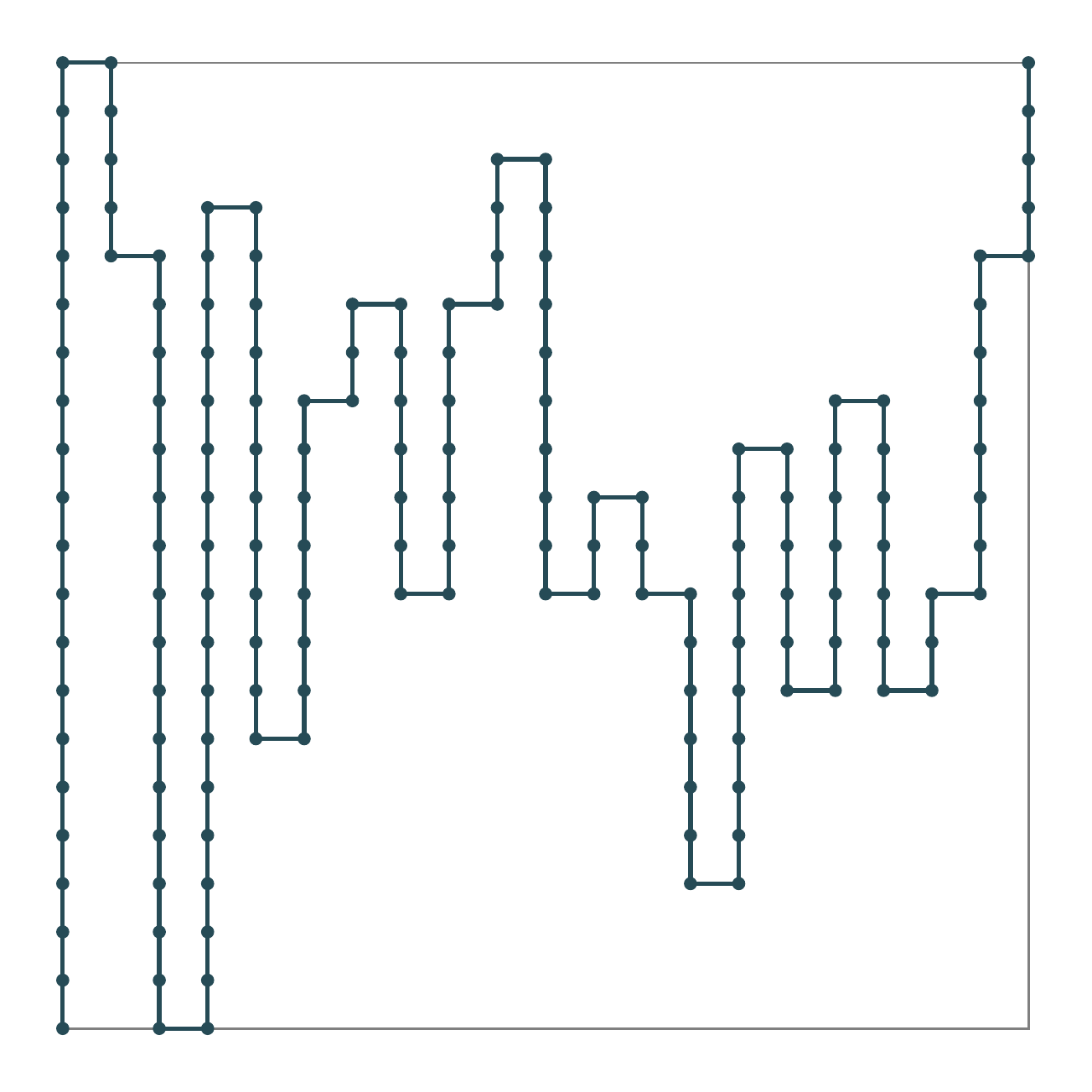}
\caption{}
\end{subfigure}
\hfill
\begin{subfigure}{0.32\textwidth}
\includegraphics[width=\textwidth]{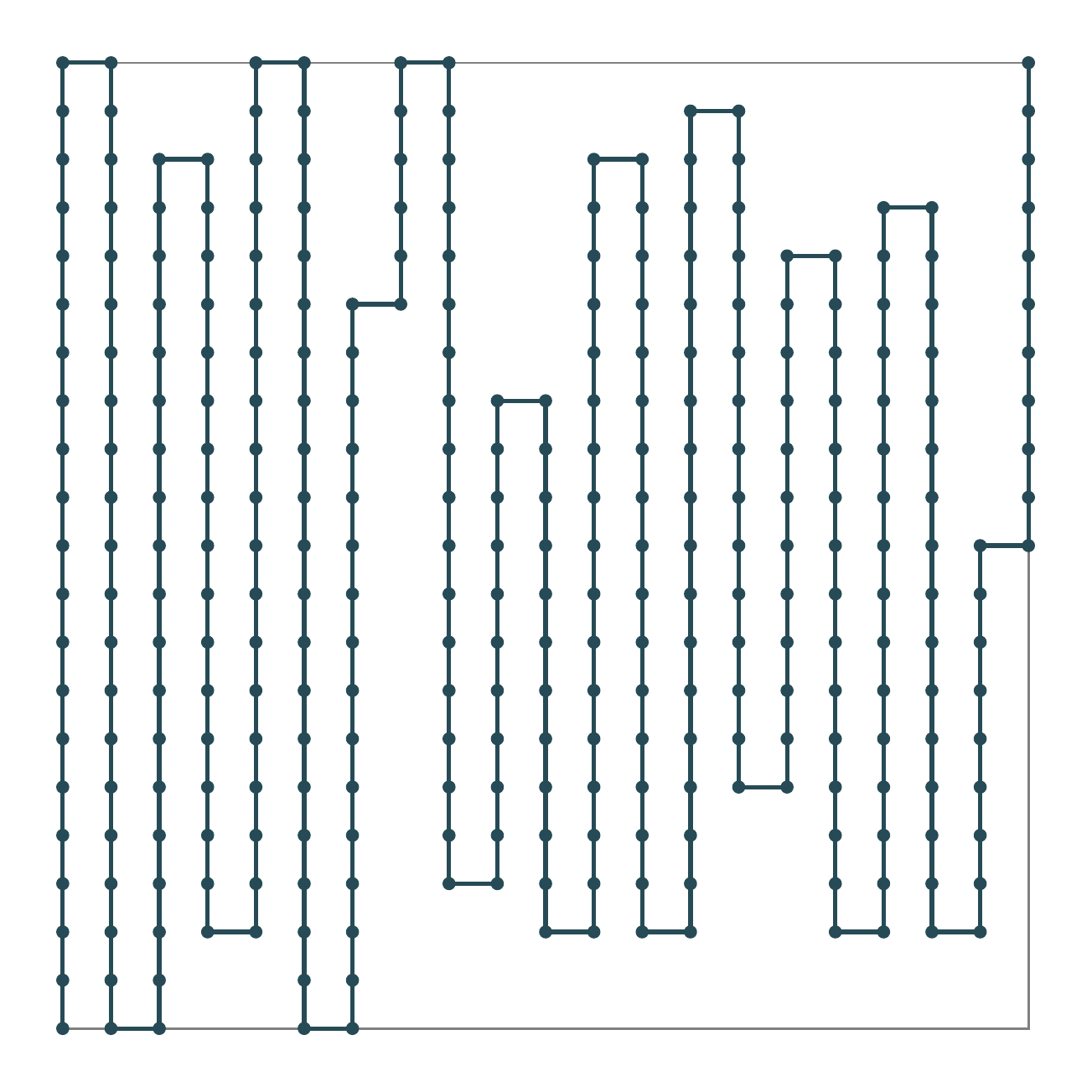}
\caption{}
\end{subfigure}
\caption{PDWs in a box of size $L=20$ sampled from the Boltzmann distribution, at (a) $t=0.8$, (b) $t=1$ and (c) $t=1.2$. The respective lengths are 92, 170 and 326.}
\label{fig:samples}
\end{figure}

\begin{thm}\label{thm:main}
The partition functions $P_L(t)$ satisfy the following.
\begin{itemize}
\item[\textup{(i)}] For $t=1$,
\begin{equation}
P_L(1) = (L+1)^L \sim e\cdot e^{L \log L}.
\end{equation}
\item[\textup{(ii)}] For $0<t<1$,
\begin{equation}\label{eqn:pf_asymp_smallt}
P_L(t) \sim \frac{1}{\sqrt{\pi}} \cdot \left(\frac{1-t^2}{1+t^2}\right)^2 \cdot L^{-1/2} \cdot \left(\frac{4t^2}{1-t^2}\right)^L.
\end{equation}
\item[\textup{(iii)}] For $t>1$,
\begin{equation}\label{eqn:pf_asymp_larget}
P_L(t) \sim \begin{cases} \displaystyle \left(\frac{t^4}{t^2-1}\right)^L t^{L^2} & L \text{ even} \\ \displaystyle \frac{t^2-1}{t^2} \cdot L^2 \cdot \left(\frac{t^3}{t^2-1}\right)^L \cdot t^{L^2} & L \text{ odd.} \end{cases}
\end{equation}
\end{itemize}
\end{thm}

See \cref{fig:pfs} for plots of $P_L(t)$ for $t=\frac12$ and $t=2$.

\begin{figure}
\centering
\begin{subfigure}{0.49\textwidth}
\includegraphics[width=\textwidth]{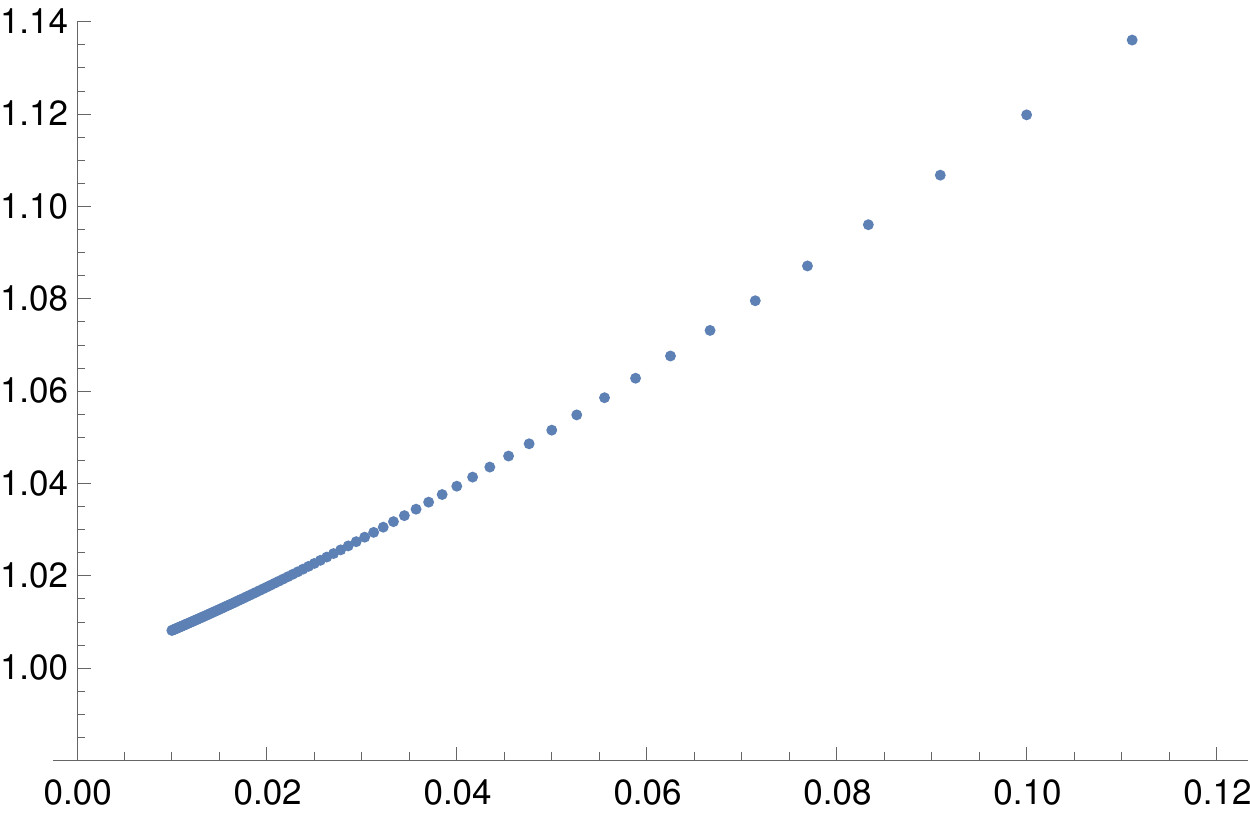}
\caption{}
\end{subfigure}

\begin{subfigure}{0.49\textwidth}
\includegraphics[width=\textwidth]{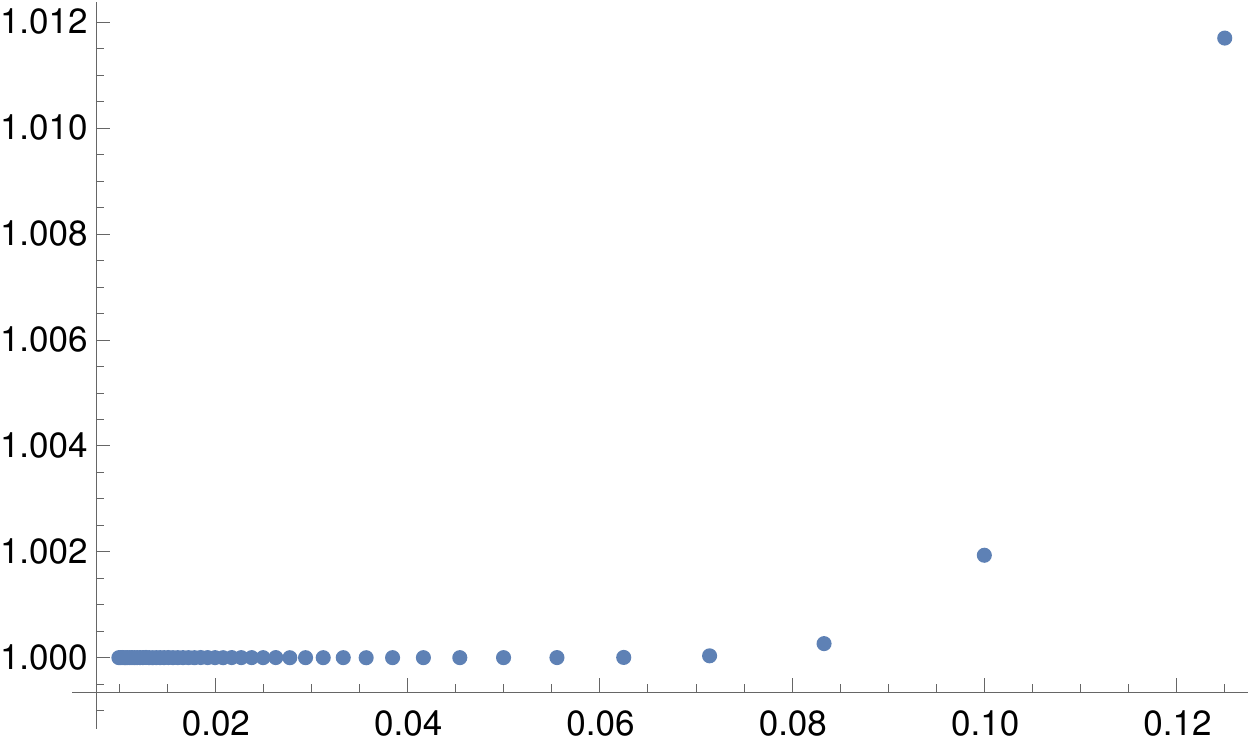}
\caption{}
\end{subfigure}
\hfill
\begin{subfigure}{0.49\textwidth}
\includegraphics[width=\textwidth]{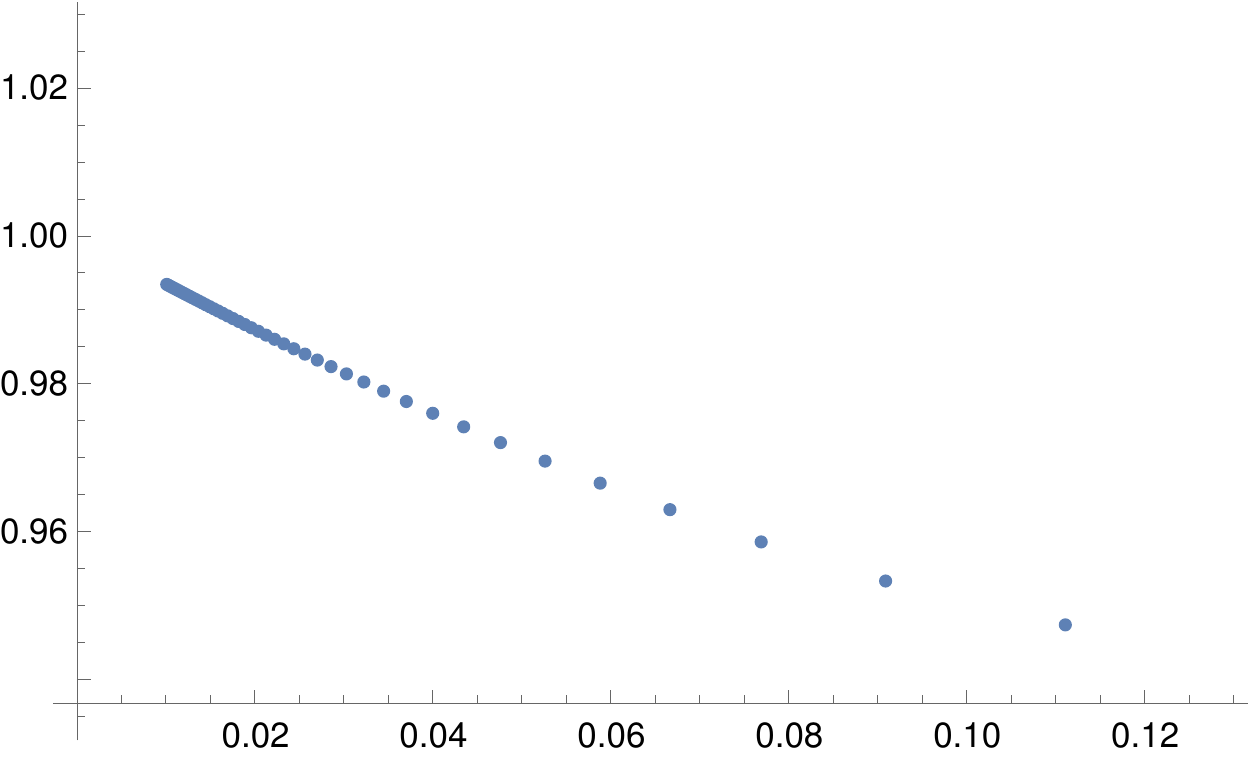}
\caption{}
\end{subfigure}
\caption{(a)~Plot of $P_L(t)$ divided by the expression~\eqref{eqn:pf_asymp_smallt} against $\frac1L$ at $t=\frac12$ for $L$ up to 100. (b)--(c) Plots of $P_L(t)$ divided by the expression~\eqref{eqn:pf_asymp_larget} against $\frac1L$ at $t=2$ for $L$ up to 100, for (b) even $L$ and (c) odd $L$.}
\label{fig:pfs}
\end{figure}

\begin{lem}\label{lem:mean_length}
The mean number of steps $\langle n \rangle_L$ satisfies the following.
\begin{itemize}
\item[\textup{(i)}] For $t=1$,
\begin{equation}
\langle n \rangle_L = \frac{L(L^2+7L+4)}{3(L+1)} \sim \frac{L^2}{3} + 2L - \frac23.
\end{equation}
\item[\textup{(ii)}] For $0<t<1$,
\begin{equation}
\langle n \rangle_L \sim \frac{2L}{1-t^2} - \frac{8t^2}{1-t^4}.
\end{equation}
\item[\textup{(iii)}] For $t>1$, 
\begin{equation}
\langle n \rangle_L \sim \begin{cases} \displaystyle L^2 + \frac{2(t^2-2)L}{t^2-1} & L \text{ even} \\
\displaystyle L^2 + \frac{(t^2-3)L}{t^2-1} + \frac{2}{t^2-1} & L \text{ odd.} \end{cases}
\end{equation}
\end{itemize}
\end{lem}

Parts (ii) and (iii) of \cref{lem:mean_length} follow by applying~\eqref{eqn:meanlength} to the respective results in \cref{thm:main}. Part (i) follows by applying~\eqref{eqn:meanlength} to~\eqref{eqn:pf_from_matrix}. See \cref{fig:expected_lengths} for plots of $\langle n \rangle_L$ for $t=\frac12$, $t=1$ and $t=2$.

\begin{figure}
\centering
\begin{subfigure}{0.49\textwidth}
\includegraphics[width=\textwidth]{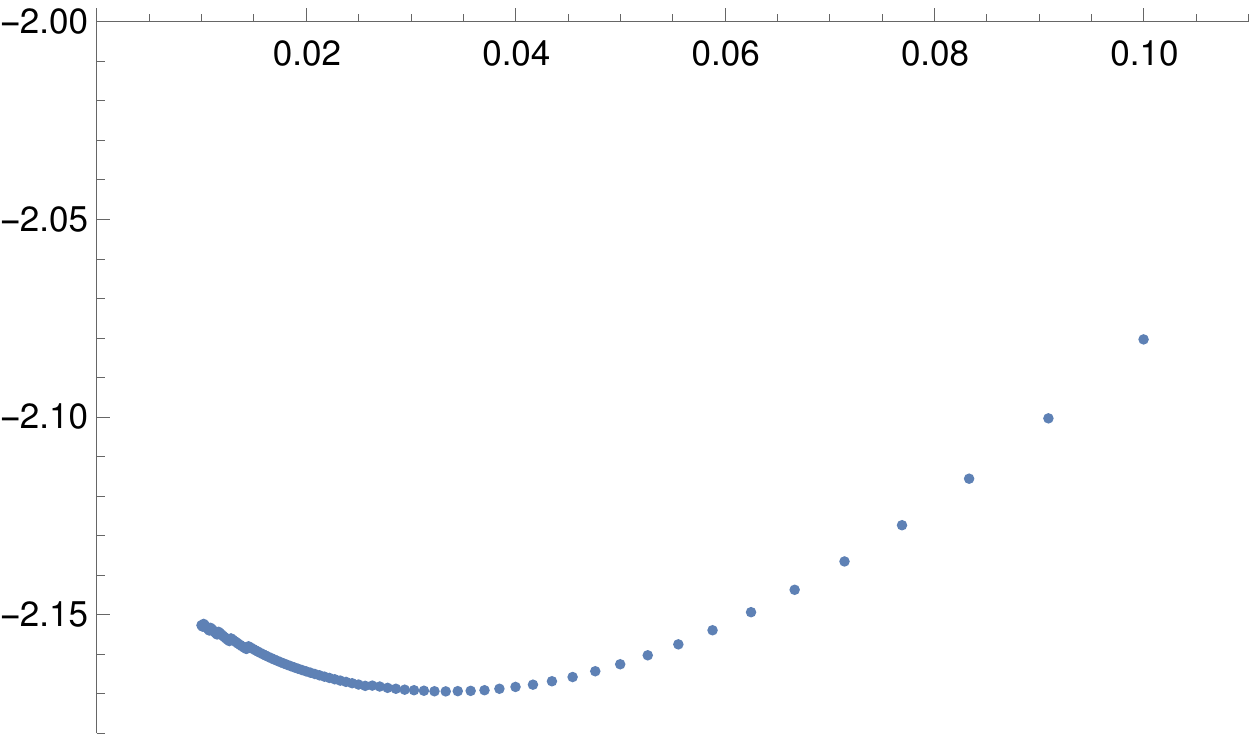}
\caption{}
\label{fig:t0p5_expected_length}
\end{subfigure}
\hfill
\begin{subfigure}{0.49\textwidth}
\includegraphics[width=\textwidth]{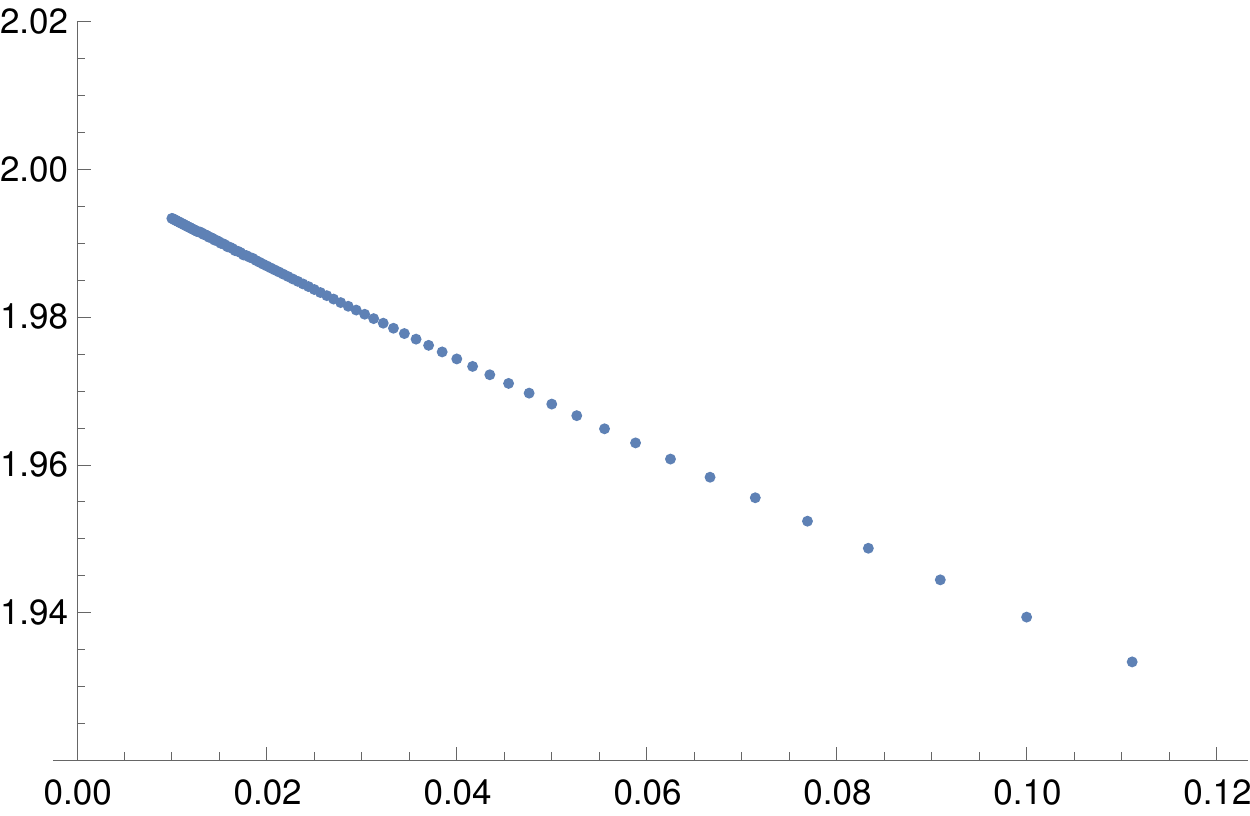}
\caption{}
\end{subfigure}

\begin{subfigure}{0.49\textwidth}
\includegraphics[width=\textwidth]{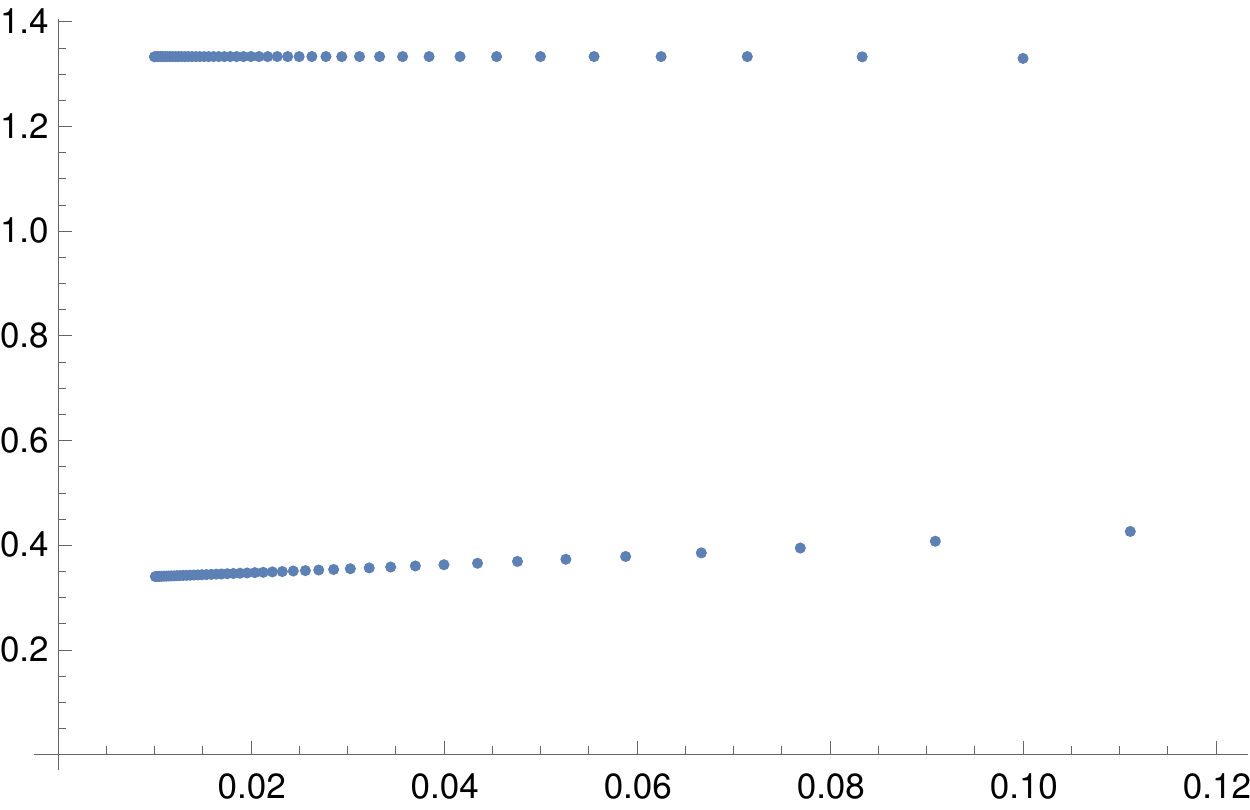}
\caption{}
\label{fig:t2_expected_length}
\end{subfigure}
\caption{(a) Plot of $\langle n\rangle_L-\frac{8L}{3}$ against $\frac1L$ at $t=\frac12$ for $L$ up to 100. The points are approaching $-\frac{8t^2}{1-t^4} = -\frac{32}{15}$. (b) Plot of $\frac1L(\langle n \rangle_L - \frac{L^2}{3})$ against $\frac1L$ at $t=1$ for $L$ up to 100. (c) Plot of $\frac1L(\langle n \rangle_L - L^2)$ against $\frac1L$ at $t=2$ for $L$ up to 100. For even $L$ the points are approaching $\frac{2(t^2-2)}{t^2-1} = \frac43$ and for odd $L$ they are approaching $\frac{t^2-3}{t^2-1} = \frac13$.}
\label{fig:expected_lengths}
\end{figure}

\section{The unweighted case \texorpdfstring{$t=1$}{t=1}}\label{sec:unweighted}

When $t=1$ we are simply interested in counting the number of PDWs in the box. Let $\mathcal{P}_L$ be the set of PDWs which cross the $L\times L$ box from the bottom left corner to the top right corner. Then there is a simple bijection between $\mathcal{P}_L$ and the set
\begin{equation}
\mathcal{W}_L := \{(w_1,w_2,\dots,w_L) \in \mathbb{Z}^L \,:\, 0 \leq w_i \leq L\},
\end{equation}
where we encode a PDW by the heights of its horizontal steps, reading left to right.

Clearly
\begin{equation}
|\mathcal{W}_L| = P_L(1) = (L+1)^L.
\end{equation}

We thus have neither $\lambda_1^L$ nor $\lambda_2^{L^2}$ growth, but instead something in between, namely
\begin{equation}
P_L(1) = e\cdot e^{L \log L} \left(1 - \frac{1}{2L} + \frac{11}{24L^2} + O(L^{-3})\right)
\end{equation}
which establishes \cref{thm:main} (i).

Note that this method is of no use when computing $\langle n \rangle_L$ at $t=1$. To do this we use the expression~\eqref{eqn:pf_from_matrix}, taking its derivative and setting $t=1$.

\section{The dilute case \texorpdfstring{$t<1$}{t<1}}\label{sec:dilute}

\subsection{Computing generating functions}

For the dilute case we will compute the generating function using the kernel method and derive the asymptotics using the saddle point method. We generalise from PDWs crossing a box to PDWs in a strip, i.e.\ $S_L = \{(x,y) \in \mathbb{Z}^2 \, : \, 0 \leq y \leq L\}$. The walks all start at $(0,0)$. We use three generating functions:
\begin{itemize}
\item $H(t,s,v) \equiv H(v)$: Counts the empty walk and walks ending with a horizontal step, with $t$ conjugate to length, $s$ conjugate to horizontal span (i.e.\ number of horizontal steps) and $v$ conjugate to the height of the endpoint.
\item $U(t,s,v) \equiv U(v)$: Counts walks ending with an up step.
\item $D(t,s,v) \equiv D(v)$: Counts walks ending with a down step.
\end{itemize}
Then by appending one step at a time, we have the functional equations
\begin{align}
	H(v) &= 1 + ts\left(H(v) + U(v) + D(v)\right), \\
	U(v) &= tv\left(H(v) + U(v)\right) - tv\left(v^L[v^L]H(v) + v^L[v^L]U(v)\right), \\
	D(v) &= t\overline{v}\left(H(v) + D(v)\right) - t\overline{v}\left([v^0]H(v) + [v^0]D(v)\right).
\end{align}
Additionally by considering the bottom and top boundaries, we have
\begin{align}
	[v^0]H(v) &= 1 + ts\left([v^0]H(v) + [v^0]D(v)\right) \\
	[v^L]H(v) &= ts\left([v^L]H(v) + [v^L]U(v)\right).
\end{align}
Combining all the above and eliminating all the $U$ and $D$ terms gives
\begin{equation}\label{eqn:main_func_eqn}
	\left(1-ts+\frac{t^2s}{t-v} - \frac{t^2sv}{1-tv}\right)H(v) = 1 - \frac{t}{t-v} + \frac{t}{t-v}H_0 - \frac{tv^{L+1}}{1-tv}H_L
\end{equation}
where $H_0 = [v^0]H(v)$ and $H_L = [v^L]H(v)$.

We now apply the \emph{kernel method} to solve this equation. The kernel is
\begin{equation}
K(t,s,v) \equiv K(v) = 1-ts+\frac{t^2s}{t-v} - \frac{t^2sv}{1-tv}
\end{equation}
which has two roots in $v$, namely
\begin{align}
	v = V &= \frac{1 - ts +t^2 + t^3s - \sqrt{(1-ts+t^2+t^3s)^2 - 4t^2}}{2t} \\
	&= t + st^2 + s^2t^3 + s^3t^4 + (s^2+s^4)t^5 + \cdots
\end{align}
and
\begin{equation}
	v = V^{-1} = t^{-1} - s - s^2t^3 - s^3t^4 - (s^2+s^4)t^5 + \cdots 
\end{equation}
Since $H(v)$ has only finite powers of $v$ (namely, $v^0$ to $v^L$), both of the kernel roots can be substituted into~\eqref{eqn:main_func_eqn} with $H(v)$ still being a well-defined (Laurent) series in $t$. We thus cancel the LHS and get a pair of equations with unknowns $H_0$ and $H_L$, which can be solved. We get
\begin{equation}\label{eqn:HL_solution}
	H_L = \frac{V^L(t-V)(1-tV)(1-V^2)}{t((t-V)^2 V^{2L+2}-(1-tV)^2)}
\end{equation}
and similar for $H_0$.

\subsection{Extracting coefficients}

We know that for any fixed $L$, $H_L$ is a rational function, though the exact way in which all the square roots cancel from~\eqref{eqn:HL_solution} is far from obvious. To get PDWs crossing a box, we want
\begin{equation}
	P_L(t) = t^{-1}[s^{L+1}]H_L.
\end{equation}
Since $H_L$ is rational, it is meromorphic in the complex $s$ plane for any real (or complex) $t$. So we have
\begin{equation}
	P_L(t) = \frac{1}{2\pi it}\oint \frac{H_L}{s^{L+2}}ds
\end{equation}
where the contour integral is a simple closed curve around the origin.

The form of \eqref{eqn:HL_solution} is not particularly conducive to computing the above contour integral. Let us rewrite it slightly as
\begin{align}
	H_L = \frac{V^L(t-V)(1-V^2)}{-t(1-tV)}\cdot \frac{1}{1-(t-V)^2V^{2L+2}/(1-tV)^2}.
\end{align}
In taking the contour integral we may assume that $|s|$ is small (the exact radius will be determined shortly) so that $|V|$ is close to $t$. Then
\begin{equation}
	\left|\frac{(t-V)^2V^{2L+2}}{(1-tV)^2}\right| \sim |s^2|t^{2L+6}
\end{equation}
for large $L$. This is thus small, and so we can approximate $H_L$ as
\begin{equation}
	H_L \sim H_L^* = \frac{V^L(t-V)(1-V^2)}{-t(1-tV)}.
\end{equation}
However, we now have a problem. $H_L$ was a rational (i.e.\ meromorphic) function but $H_L^*$ is not. So there may now be branch cuts to contend with. These arise from the square root term in $V$, which is
\begin{equation}
	\sqrt{(1-ts+t^2+t^3s)^2 - 4t^2}
\end{equation}
The term inside the square root is 0 at 
\begin{equation}
	s_1=\frac{1-t}{t(1+t)} \qquad \text{and} \qquad s_2=\frac{1+t}{t(1-t)}.
\end{equation}
We have $0 < s_1 < s_2$ for $t\in(0,1)$, with $s_1 \to 0$ as $t\to 1$. See \cref{fig:s1_s2}.

\begin{figure}[h]
	\centering
    \begin{subfigure}{0.49\textwidth}
	\includegraphics[width=\textwidth]{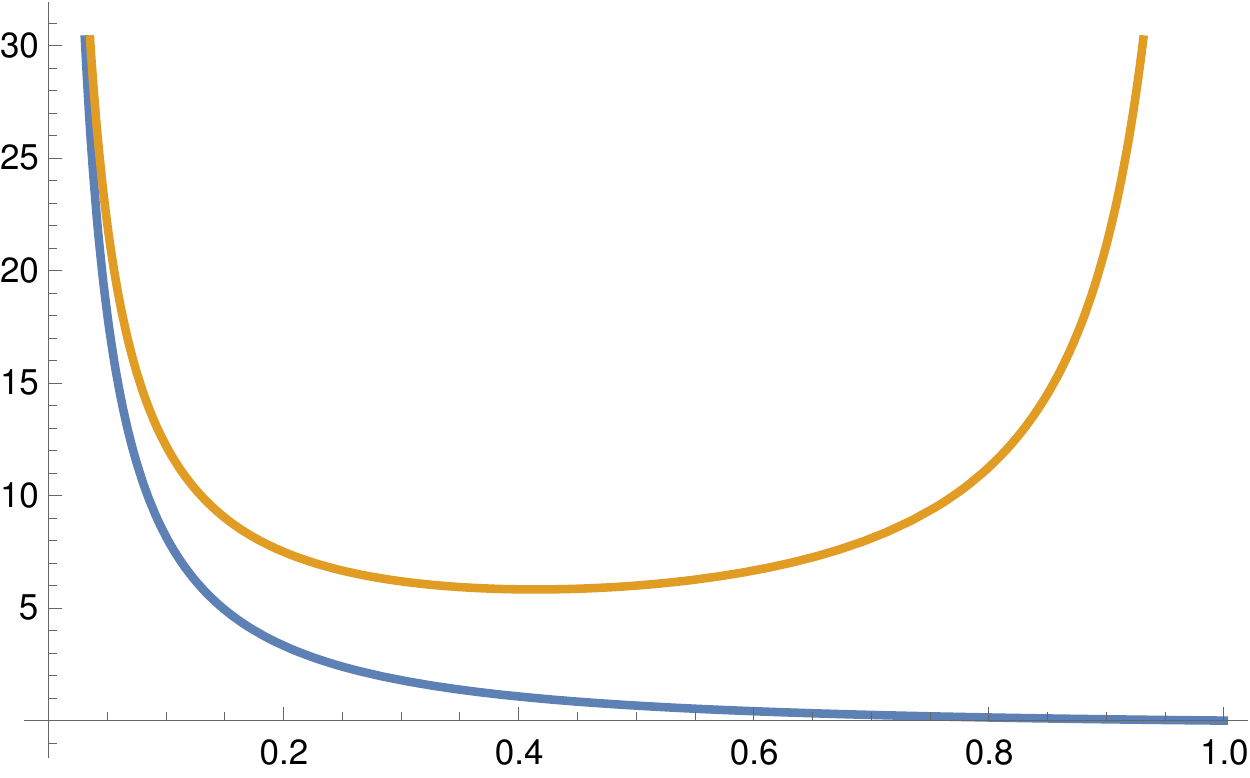}
 \caption{}
 \label{fig:s1_s2}
 \end{subfigure}
 \hfill
 \begin{subfigure}{0.49\textwidth}
 \includegraphics[width=\textwidth]{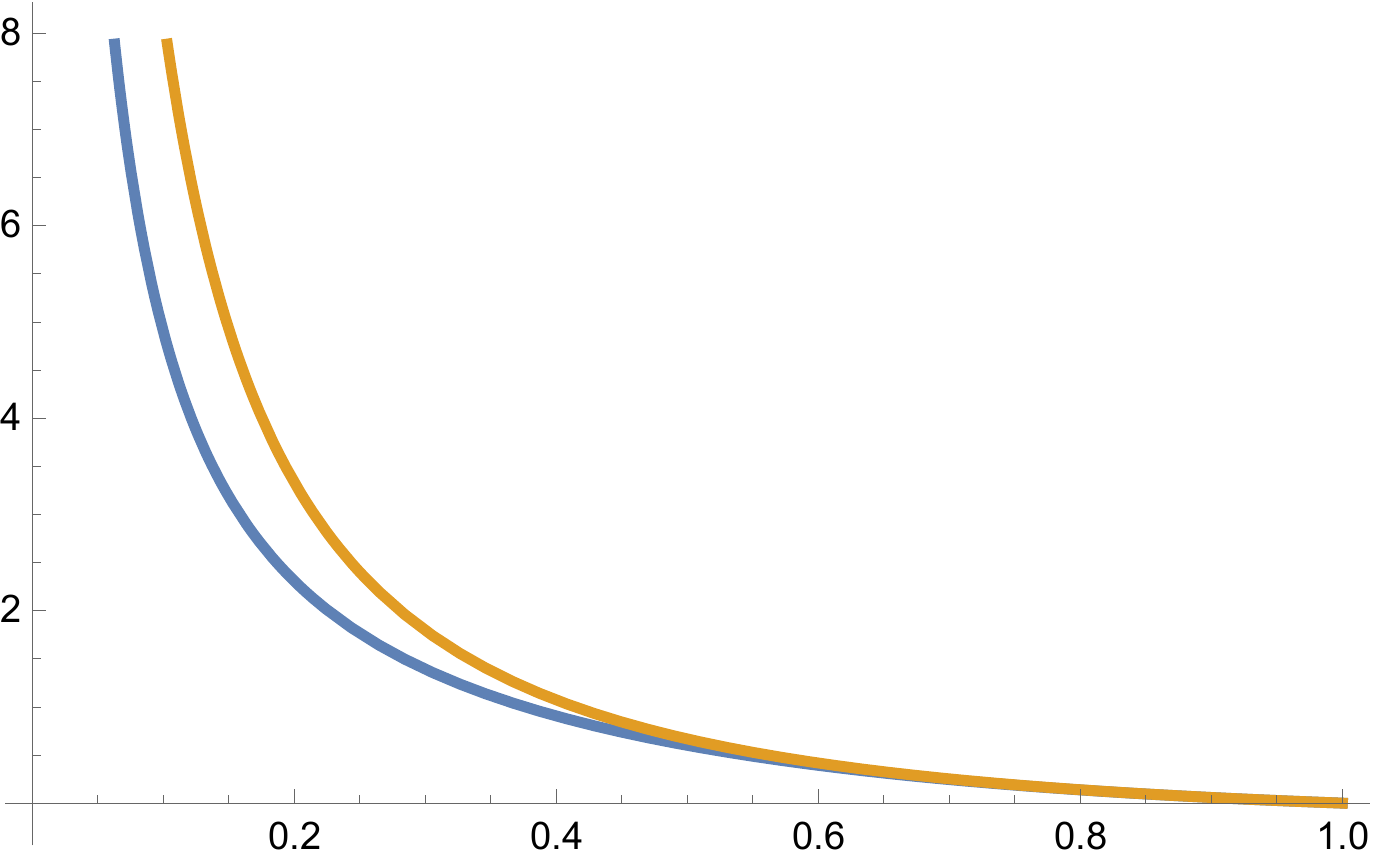}
 \caption{}
 \label{fig:s0_s1}
 \end{subfigure}
\caption{(a) Plot of $s_1$ (blue) and $s_2$ (orange). (b) Plot of $s_0$ (blue) and $s_1$ (orange).}	
\end{figure}

The term inside the square root is negative for $s_1 < s < s_2$ and positive (for real $s$) for $s<s_1$ and $s>s_2$. We may thus place the branch cut along the real axis between $s_1$ and $s_2$, and as long as our contour integral is along a curve with $|s| < s_1$ then we avoid the branch cut.

Next we need to check if $H^*_L$ has any poles that we need to take into consideration. From the form of $V$ we can see that the numerator presents no problem. For the denominator we need to check only $(1-tV)$, but a bit of rearranging shows that this has no roots in $s$.

So it remains to compute the asymptotics of
\begin{equation}\label{eqn:the_contour_integral}
	P_L^*(t) = \frac{1}{2\pi it}\oint \frac{H_L^*}{s^{L+2}}ds = \frac{1}{2\pi it} \oint \frac{V^L(t-V)(1-V^2)}{-t(1-tV)s^{L+2}} ds
\end{equation}
where the contour has to be within $|s| < s_1$.

\subsection{Asymptotics via the saddle point method}

The most basic form of the saddle point method gives
\begin{equation}
	\int g(z) \exp(n h(z)) dz \sim i\sqrt{\frac{2\pi}{n h''(z_0)}} g(z_0) \exp(nh(z_0)), \qquad n\to\infty
\end{equation}
where $z_0$ is a saddle point of $h(z)$.

The form~\eqref{eqn:the_contour_integral} is well set up for estimation using the saddle point method. The dependence on $L$ is from
\begin{equation}
	\left(\frac{V}{s}\right)^L = \exp(Lh(s))
\end{equation}
where $h(s) = \log V - \log s$. $h$ has a saddle point at 
\begin{equation}
	s_0 = \frac{1-t^2}{2t(1+t^2)}
\end{equation}
It is straightforward to check that $0 < s_0 < s_1$ for $0<t<1$ (see \cref{fig:s0_s1}). Both $s_0, s_1 \to 0$ as $t\to1$. 


For us
\begin{equation}
	g(s) = \frac{(t-V)(1-V^2)}{-t(1-tV)s^2}
\end{equation}
Substituting,
\begin{equation}
	g(s_0) = 4t^2.
\end{equation}
Meanwhile
\begin{align}
	\exp(h(s_0)) &= \frac{4t^2}{1-t^2} \\
	h''(s_0) &= \frac{8t^2(1+t^2)^4}{(1-t^2)^4}
\end{align}
Putting this all together,
\begin{align}
	P_L^*(t) &\sim \frac{1}{2\pi i t} \cdot i 4t^2 L^{-1/2} \sqrt{2\pi \cdot \frac{(1-t^2)^4}{8t^2(1+t^2)^4}} \left(\frac{4t^2}{1-t^2}\right)^L \\
	&= \frac{1}{\sqrt{\pi}} \cdot \left(\frac{1-t^2}{1+t^2}\right)^2 \cdot L^{-1/2} \cdot \left(\frac{4t^2}{1-t^2}\right)^L \label{eqn:GL_final_approx}
\end{align}
as in \cref{thm:main} (ii).

\section{The dense case \texorpdfstring{$t>1$}{t>1}}\label{sec:dense}

\subsection{Transfer matrix formulation and Bethe ansatz solution}

For the dense case we must use a completely different method to compute asymptotics, using a transfer matrix approach. Define the $(L+1)\times(L+1)$ matrix
\begin{equation}
T_L(t) = \begin{pmatrix} t & t^2 & t^3 & \cdots & t^{L+1} \\ t^2 & t & t^2 & \cdots & t^L \\ t^3 & t^2 & t & \cdots & t^{L-1} \\ \vdots & \vdots & \vdots & \ddots & \vdots \\ t^{L+1} & t^L & t^{L-1} & \cdots & t \end{pmatrix}
\end{equation}
Then
\begin{align}
	P_L(t) &= (1, t, t^2, \dots, t^L) \cdot T_L(t)^L \cdot (0, 0, \dots, 0, 1)^\mathsf{T} \label{eqn:pf_from_matrix}\\
	&= \frac1t(1,0,0,\dots,0) \cdot T_L(t)^{L+1} \cdot (0,0,\dots,0,1)^\mathsf{T}.
\end{align}

For brevity, in the following  we may drop subscripts or functional arguments. Let us consider the eigen-equation
\begin{equation}
T\mathbf{g} = \lambda\mathbf{g},
\end{equation}
that is
\begin{equation}
\sum_{j=1}^{L+1}T_{i,j}g_j = \lambda g_i \qquad \text{for } i=1,\dots,L+1.
\end{equation}
We begin with the ansatz $g_j = z^j$ for some complex number $z$, giving
\begin{equation}
\sum_{j=1}^{L+1}t^{|i-j|+1}z^j = \lambda z^i.
\end{equation}
Splitting the sum gives
\begin{equation}
\sum_{j=1}^i t^{(i-j)+1}z^j + \sum_{j=i+1}^{L+1} t^{(j-i)+1}z^j = \lambda z^i
\end{equation}
or rather
\begin{equation}
zt^i\sum_{k=0}^{i-1}\left(\frac{z}{t}\right)^k + t^{1-i}(tz)^{i+1}\sum_{k=0}^{L-i}(tz)^k = \lambda z^i.
\end{equation}
Summing the partial geometric series, we find
\begin{equation}
zt^i\left[\frac{1-(\frac{z}{t})^i}{1-\frac{z}{t}}\right] + z^{i+1}t^2\left[\frac{1-(tz)^{L-i+1}}{1-tz}\right] = \lambda z^i.
\end{equation}
Collecting terms gives
\begin{equation}
t\left[\frac{t^{i-1}z}{1-\frac{z}{t}} - \frac{t^{L+2-i}z^{L+2}}{1-tz}\right] + tz^i\left[\frac{tz}{1-tz}-\frac{\frac{z}{t}}{1-\frac{z}{t}}\right] = \lambda z^i.
\end{equation}
Since this needs to hold for all $i$, we obtain the eigenvalue $\lambda$ as
\begin{equation}
\lambda = \lambda(t,z) = t\left[\frac{tz}{1-tz}-\frac{\frac{z}{t}}{1-\frac{z}{t}}\right] = -\frac{z(1-t^2)}{(1-tz)(1-\frac{z}{t})} = \frac{t(1-t^2)}{1-t(z+\frac{1}{z})+t^2}.
\end{equation}

We immediately note that
\begin{equation}
\lambda(t,z) = \lambda(t,\textstyle\frac{1}{z})
\end{equation}
and so to remove the boundary terms we extend the ansatz to
\begin{equation}
g_j = z^j + C(t,z)z^{-j}.
\end{equation}
The same $\lambda$ as above still works, and cancels the $z^i$ and $z^{-i}$ terms. We are left with the boundary equation
\begin{equation}
	t\left[ \frac{t^{i-1}z}{1- \frac{z}{t}} - \frac{t^{L+2-i}z^{L+2}}{1-tz} \right] + C t\left[ \frac{t^{i-1}}{z(1- \frac{1}{tz})} - \frac{t^{L+2-i}z^{-(L+2)}}{1-\frac{t}{z}} \right] =0 
\end{equation}
which after multiplying by $t^{i-1}$ we rewrite as
\begin{equation}
	\left[ \frac{t^{2i-1}z}{1- \frac{z}{t}} - \frac{t^{L+2}z^{L+2}}{1-tz} \right] - C \left[ \frac{t^{2i}}{1- tz} + \frac{t^{L+2}z^{-(L+2)}}{1-\frac{t}{z}} \right] =0 
\end{equation}
that is
\begin{equation}
	t^{2i}\left[ \frac{\frac{z}{t}}{1- \frac{z}{t}} - C \frac{1}{1- tz} \right] -  t^{L+2}\left[ \frac{z^{L+2}}{1-tz}  + C \frac{z^{-(L+2)}}{1-\frac{t}{z}} \right] =0.
\end{equation}
This must hold for each $i$ so we expect each term to be zero individually. We seek to set $C$ so that there is a common factor between the two, which can then be cancelled by $z$. Comparing the two terms, we see that $C=\pm z^{L+2}$ will make them the same, up to a simple factor. First with $C=z^{L+2}$, the above becomes
\begin{equation}
	t^{2i}\alpha(t,z) + t^{L+2}\alpha(t,z) = 0
\end{equation}
where
\begin{align}
	\alpha(t,z) &= \frac{\frac{z}{t}}{1-\frac{z}{t}}-\frac{z^{L+2}}{1-zt} \\ 
	&= -\frac{1}{(1-tz)(1-\frac{t}{z})}\cdot (1-tz-tz^{L+1}+z^{L+2}).
\end{align}
On the other hand with $C=-z^{L+2}$, we get
\begin{equation}
	t^{2i}\beta(t,z) - t^{L+2}\beta(t,z) = 0
\end{equation}
where
\begin{align}
	\beta(t,z) &= \frac{\frac{z}{t}}{1-\frac{z}{t}} + \frac{z^{L+2}}{1-zt} \\ 
	&= -\frac{1}{(1-tz)(1-\frac{t}{z})}\cdot (1-tz+tz^{L+1}-z^{L+2}).
\end{align}

The above thus  gives that the eigenvectors $\mathbf{g}_{L,k}$ of $T_L$, where $k=1,\dots,L+1$, are of the form
\begin{equation}\label{eqn:eigenvectors}
g_{L,k,j} = z^j +(-1)^{k+1} z^{L+2-j} \text{ with } j=1,\dots,L+1
\end{equation}
where the $z=z_{L,k}$ are complex numbers. Specifically, the $z_{L,k}$ are roots of the polynomials $A_{L,k}(t,z)$, which combine $\alpha$ and $\beta$ from above:
\begin{equation}
A_{L,k}(t,z) = 1-tz+(-1)^k z^{L+1}(t-z).
\end{equation}
Note that
\begin{equation}\label{eqn:A_palindromic}
z^{L+2}A_{L,k}(t,\textstyle\frac1z) = (-1)^{k+1}A_{L,k}(t,z)
\end{equation}
so that if $z$ is a root then so too is $\frac1z$. The property~\eqref{eqn:A_palindromic} makes $A_{L,k}$ a \emph{self-inversive} polynomial, and in particular it is \emph{palindromic} for odd $k$, and \emph{antipalindromic} for even $k$. Since the roots come in reciprocal pairs, in the following $z_{L,k}$ can refer to either representative of a pair (it will make no difference which value is chosen).

Next, we observe that $A_{L,k}$ is of degree $L+2$, however
\begin{itemize}
\item when $L,k$ are both odd, $A_{L,k}(t,-1)=0$, but then at $z=-1$ we have $g_{L,k,j}=0$ for all $j$,
\item when $L$ is odd and $k$ is even, $A_{L,k}(t,1) = 0$, but then at $z=1$ we again have $g_{L,k,j}=0$,
\item when $L,k$ are both even, $A_{L,k}(t,1) = A_{L,k}(t,-1) = 0$, but then at $z=\pm1$ we again have $g_{L,k,j} = 0$.
\end{itemize}
(Note that $A_{L,k}$ never has a double pole at $z=\pm 1$, which is easily seen by checking derivatives.) The roots at $z=\pm1$ are thus trivial and are not counted among the $z_{L,k}$. Factoring out the trivial $(1\pm z)$ terms then gives the polynomials
\begin{align}
B_{L,k}(t,z) &= 1+(1+t)\sum_{n=1}^L(-1)^nz^n + z^{L+1} & & L,k \text{ odd} \\
B_{L,k}(t,z) &= 1 + (1-t)\sum_{n=1}^L z^n+ z^{L+1} & & L\text{ odd, } k \text{ even} \\
B_{L,k}(t,z) &= 1 -(t-z) \sum_{\substack{n=1 \\ n \text{ odd}}}^{L-1}z^n & & L,k \text{ even} \\
B_{L,k}(t,z) &= 1-tz-tz^{L+1} + z^{L+2} & & L\text{ even, } k \text{ odd}
\end{align}
whose roots are exactly the reciprocal pairs $z_{L,k}$. It is easy to check that each of the $B_{L,k}(t,z)$ are palindromic. Indeed, setting $z = \frac{1}{z}$ in~\eqref{eqn:eigenvectors} leads to\begin{equation}
z^{L+2}\left[\mathbf{g}_{L,k}\right]_{z=\frac1z} = (-1)^{k+1}\mathbf{g}_{L,k}
\end{equation}
so that each reciprocal pair of roots gives the same eigenvalue / vector pair. 

Next we diagonalise (using the fact that $T_L$ is real symmetric), to get
\begin{equation}
G_L(t) = \frac1t (1,0,\dots,0) \cdot \left(\sum_{k=1}^{L+1}\tilde{\mathbf{g}}_{L,k}^\mathsf{T}\lambda_{L,k}^{L+1}\tilde{\mathbf{g}}_{L,k}\right)\cdot(0,\dots,0,1)^\mathsf{T}
\end{equation}
where
\begin{equation}
\tilde{\mathbf{g}}_{L,k} = \frac{\mathbf{g}_{L,k}}{\lVert \mathbf{g}_{L,k}\rVert}.
\end{equation}
Now
\begin{align}
\lVert\mathbf{g}_{L,k}\rVert^2 &= \sum_{j=1}^{L+1}(z_{L,k}^j + (-1)^{j+1}z_{L,k}^{L+2-j})^2 \\
&= \frac{2z_{L,k}^2(1+(-1)^{k+1}(L+1)z_{L,k}^L(1-z_{L,k}^2)-z_{L,k}^{2L+2})}{1-z_{L,k}^2}.
\end{align}
Substituting,
{\small
\begin{align}
P_L(t) &= \frac1t (1,0,\dots,0) \cdot \left(\sum_{k=1}^{L+1} \frac{\mathbf{g}_{L,k}^{\mathsf{T}}\lambda_{L,k}^{L+1}\mathbf{g}_{L,k}}{\lVert \mathbf{g}_{L,k}\rVert^2}\right) \cdot (0,\dots,0,1)^\mathsf{T} \\
&= \frac1t (1,0,\dots,0)\cdot\left(\sum_{k=1}^{L+1}\mathbf{g}_{L,k}^\mathsf{T}\mathbf{g}_{L,k}\frac{(1-z_{L,k}^2)\lambda_{L,k}^{L+1}}{2z_{L,k}^2(1+(-1)^{k+1}(L+1)z_{L,k}^L(1-z_{L,k}^2)-z_{L,k}^{2L+2})}\right)\cdot(0,\dots,0,1)^\mathsf{T} \\
&= \frac1t \sum_{k=1}^{L+1}\frac{g_{L,k,1}g_{L,k,L+1}(1-z_{L,k}^2)\lambda_{L,k}^{L+1}}{2z_{L,k}^2(1+(-1)^{k+1}(L+1)z_{L,k}^L(1-z_{L,k}^2)-z_{L,k}^{2L+2})} \\
&= \frac{(1-t^2)^{L+1}}{2t}\sum_{k=1}^{L+1} \frac{(-1)^{k+1}(1+(-1)^{k+1}z_{L,k}^L)^2(1-z_{L,k}^2)z_{L,k}^{L+1}}{(1+(-1)^{k+1}(L+1)z_{L,k}^L(1-z_{L,k}^2)-z_{L,k}^{2L+2})(z_{L,k}-t)^{L+1}(\frac1t-z_{L,k})^{L+1}}. \label{eqn:GL_sum_formulation}
\end{align}
}
We again note that the above sum is over the $L+1$ reciprocal pairs of roots, and for each $k$ it does not matter which of the pair is chosen. In the following subsection we will make things more explicit.

\subsection{The roots for \texorpdfstring{$t>1$}{t>1}}\label{sec:roots}

The asymptotics of~\eqref{eqn:GL_sum_formulation} depend on the values of the complex numbers $z_{L,k}$. There are $L+1$ (pairs) of these; however, it turns out that for $t>1$ only two of them contribute to the dominant asymptotics. This is partly because of the following remarkable fact.

\begin{lem}\label{lem:real_roots}
For $t>1$ and $L > \frac{2}{t-1}$, $L-1$ of the reciprocal pairs of roots $z_{L,k}$ are on the unit circle, and two pairs (one for even $k$ and one for odd $k$) are real, positive and not on the unit circle.
\end{lem}

See \cref{fig:t1p2_zk} for an illustration at $t=\frac65$.

\begin{figure}
\centering
\includegraphics[width=0.5\textwidth]{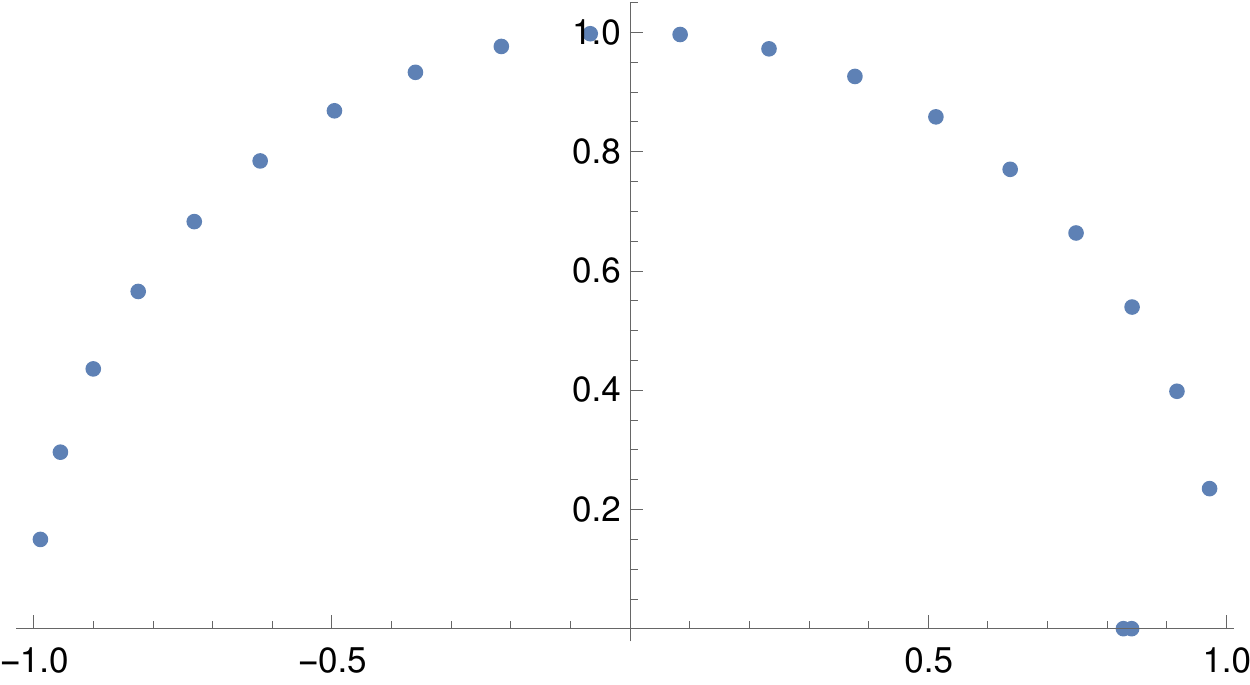}
\caption{The roots $z_{20,k}$ in the complex plane at $t=\frac65=1.2$. For the reciprocal pairs on the unit circle we have chosen those with positive imaginary part, and for those on the real line we have chosen those inside the unit circle. Note the two real roots $z_{20,1}$ and $z_{20,2}$ close to $\frac1t = \frac56$.}
\label{fig:t1p2_zk}
\end{figure}

We will make use of a result due to Vieira. First, we make precise a term we used in the previous subsection. A polynomial 
\begin{equation}
p(z) = a_0 + a_1 z+ \dots + a_n z^n
\end{equation}
with coefficients in $\mathbb{C}$ and with $a_n \neq 0$ is \emph{self-inversive} if it satisfies
\begin{equation}
p(z) = \omega z^n \overline{p}(\textstyle \frac1z)
\end{equation}
with $|\omega| = 1$, where $\overline{p}(z)$ is the complex conjugate of $p(z)$.

\begin{lem}[\cite{vieira_number_2017}]\label{lem:vieira}
Let $p(z) = a_0 + a_1z + \dots + a_nz^n$ be a self-inversive polynomial of degree $n$. If
\begin{equation}\label{eqn:vieira}
|a_{n-l}| > \frac12 \sum_{\substack{k=0 \\ k \neq l, n-l}}^n |a_k|, \qquad \qquad l < \frac{n}{2},
\end{equation}
then $p(z)$ has at least $n-2l$ roots on the unit circle.
\end{lem}

\begin{proof}[Proof of \cref{lem:real_roots}]
The polynomials $A_{L,k}(t,z)$ and $B_{L,k}(t,z)$ are self-inversive with $\omega = (-1)^{k+1}$ and $\omega=1$ respectively. For now it is simpler to work with the $A_{L,k}$, keeping in mind the two trivial roots at $z=\pm1$.

Take $p(z) = A_{L,k}(t,z)$ and set $l=1$ in \cref{lem:vieira}. Then the condition~\eqref{eqn:vieira} is simply $t>1$, so at least $L$ of the $L+2$ roots of $A_{L,k}$ are on the unit circle (note that these include the trivial roots), i.e.\ at most two are not on the unit circle. It remains to show that exactly two are not on the unit circle, both for odd and even $k$.

For odd $k$, any root satisfies
\begin{equation}
z^{L+1} = \frac{1-tz}{t-z} = m(z).
\end{equation}
For $z\in(0,1)$ we clearly have that $z^{L+1}$ is a strictly increasing function $(0,1) \to (0,1)$. On the other hand
\begin{equation}
m'(z) = - \frac{t^2-1}{(t-z)^2}
\end{equation}
so $m(z)$ is a strictly decreasing function mapping $(0,1)$ to $(\frac1t,-1)$. It follows that there must be a root $z=z_{L,1} \in (0,1)$. Since $A_{L,k}$ is self-inversive, there is another root at $z=\frac{1}{z_{L,1}} > 1$

For even $k$, any root satisfies $z^{L+1} = -m(z)$. Now the RHS is a strictly increasing function $(0,1) \to (-\frac1t,1)$. To establish the existence of a root, note that
\begin{equation} 
\frac{d}{dz}z^{L+1} = (L+1)z^L \to (L+1) \text{ as } z\to 1^-
\end{equation}
while
\begin{equation}
\frac{d}{dz}(-m(z)) = \frac{t^2-1}{(t-z)^2} \to \frac{t+1}{t-1} \text{ as } z\to 1^-.
\end{equation}
Thus if
\begin{equation}
L+1 > \frac{t+1}{t-1} \quad \iff \quad L > \frac{2}{t-1}
\end{equation}
then as $z\to1^-$, $z^{L+1}$ approaches 1 at a greater slope than $-m(z)$, and hence $z^{L+1} < -m(z)$ for $z \in (1-\epsilon,1)$ for some $\epsilon>0$. So there is a real root $z=z_{L,2} \in (0,1)$. Again by the self-inversive property, there must be another at $z=\frac{1}{z_{L,2}} > 1$.
\end{proof}

It is the two roots $z_{L,1}$ and $z_{L,2}$ inside the unit circle which are now of interest, and the next step is to compute the asymptotic behaviour of these as $L$ grows large. First, observe that
\begin{equation}
A_{L,k}(t,\textstyle\frac1t) = (-1)^k(\textstyle\frac1t)^{L+1}(t+\textstyle\frac1t) \to 0 \qquad \text{as } L \to \infty,
\end{equation}
while for $\epsilon>0$ with $0 < \frac1t -\epsilon < \frac1t + \epsilon < 1$ we have
\begin{equation}
A_{L,k}(t,\textstyle\frac1t \pm \epsilon) \to \mp \epsilon t \qquad \text{as } L\to\infty.
\end{equation}
It follows that the two roots $z_{L,1}$ and $z_{L,2}$ must approach $\frac1t$ as $L\to\infty$. Next, rearrange the equation $A_{L,k}(t,z) = 0$ to get
\begin{equation}\label{eqn:Aeqn_rewritten}
z^{L+1} = (-1)^{k+1}\frac{1-tz}{t-z}.
\end{equation}
This implies that $z_{L,1} < \frac1t$ while $z_{L,2} > \frac1t$. Rearranging again,
\begin{align}
\log\left[(-1)^{k+1}(\textstyle\frac1t-z)\right] &= (L+1)\log z + \log t+\log(t-z) \label{eqn:zL1_eqn_from_PL1} \\
&\sim -(L+1)\log t + \log t + \log(t-\textstyle\frac1t) \\
&= \log\left(\frac{t^2-1}{t^{L+3}}\right).
\end{align}
Hence for $k=1,2$,
\begin{equation}
z_{L,k} \sim z^*_{L,k} =  \frac1t + (-1)^k\left(\frac{t^2-1}{t^3}\right)t^{-L}.
\end{equation}
It will turn out that the precision of these estimates is sufficient for even $L$ but not enough for odd $L$ (this is because there is significant cancellation between the $k=1$ and $2$ terms of~\eqref{eqn:GL_sum_formulation} for odd $L$). However, we can compute more a precise estimate for $z_{L,1}$ by iterating~\eqref{eqn:zL1_eqn_from_PL1}. That is, we substitute $z_{L,k}^*$ into the RHS of~\eqref{eqn:zL1_eqn_from_PL1}. Taking the next-to-leading term then gives
\begin{equation}
z_{L,k} \sim z_{L,k}^{**} = \frac1t +(-1)^k t^{-L}\left(\frac{t^2-1}{t^3}\right)\left(1 +(-1)^k t^{-L}L\frac{t^2-1}{t^2}\right).
\end{equation}

\subsection{Asymptotics}

We will compute the leading asymptotics for $P_L(t)$ by taking only the $k=1$ and $2$ terms from~\eqref{eqn:GL_sum_formulation}. We will then need to show that the remaining terms in the sum do not contribute to the dominant asymptotics, which amounts to showing that the first factor in the denominator of the summands is not too close to 0.

\subsubsection{Even \texorpdfstring{$L$}{L}}

We take only the $k=1,2$ terms of~\eqref{eqn:GL_sum_formulation}. Any term of the form $z_{L,k}^L$ or similar approaches 0 very quickly, so for the purposes of asymptotics these are all set to 0, except for the factor of $z_{L,k}^{L+1}$ in the numerator. This, and the other $z_{L,k}$ terms except for the important $(\frac1t - z_{L,k})^{L+1}$ term in the denominator, are then set to $\frac1t$. This yields
\begin{align}
&\sim \frac{(1-t^2)^{2M+1}}{2t} \sum_{k=1}^2 \frac{(-1)^{k+1}(1-\frac{1}{t^2})t^{-L-1}}{(\frac1t-t)^{L+1}(\frac1t-z_{L,k})^{L+1}} \\
&= \frac{t^2-1}{2t^3}\left(\frac{1}{(\frac1t-z_{L,1})^{L+1}} - \frac{1}{(\frac1t-z_{L,2})^{L+1}}\right)
\end{align}
Now using the approximations $z_{L,k}^*$ this simplifies to
\begin{equation}
\sim \left(\frac{t^4}{t^2-1}\right)^L t^{L^2}, \qquad L \text{ even}.
\end{equation}

\subsubsection{Odd \texorpdfstring{$L$}{L}}

If we follow the same procedure as above but take $L$ to be odd then everything cancels and we just get $0$. So we must instead switch to the more precise estimates $z_{L,k}^{**}$. Substituting, we get
\begin{align}
&\sim \frac{t^{2L^2+7L+2}}{2(t^2-1)^L}\left(\frac{1}{(t^{L+2} - L(t^2-1))^{L+1}} - \frac{1}{(t^{L+2} + L(t^2-1))^{L+1}}\right) \\
&\sim \frac{t^{2L^2+7L+2}}{2(t^2-1)^L} \cdot \frac{1}{t^{L(L+2)}}\left(\frac{1}{t^{L+2}-L(L+1)(t^2-1)} - \frac{1}{t^{L+2}+L(L+1)(t^2-1)}\right) \\
&= \frac{t^{L^2+3L-2}}{(t^2-1)^L} \cdot \frac{L(L+1)(t^2-1)}{1-\frac{L^2(L+1)^2(t^2-1)^2}{t^{2L+4}}} \\
&\sim \frac{t^2-1}{t^2} \cdot L^2 \cdot \left(\frac{t^3}{t^2-1}\right)^L \cdot t^{L^2}, \qquad L \text{ odd}.
\end{align}
To get from the first to the second line above we have used $(1+x)^L \sim 1+Lx$ for each of the two terms in the large parentheses.

\subsubsection{The roots on the unit circle}

With factors of the form $t^{L^2}$ coming from the $k=1$ and $2$ terms in the sum~\eqref{eqn:GL_sum_formulation}, the remaining terms can only affect the dominant asymptotics if the factor
\begin{equation}
D_{L,k}(t) = 1+(-1)^{k+1}(L+1)z_{L,k}^L(1-z_{L,k}^2)-z_{L,k}^{2L+2}
\end{equation}
in the denominator is very close to 0. Here we show this is not the case. Firstly,~\eqref{eqn:Aeqn_rewritten} can be used to eliminate the $z_{L,k}^L$ and $z_{L,k}^{2L+2}$ terms, giving
\begin{equation}
D_{L,k}(t) = 1+(L+1)(1-z_{L,k}^2)\frac{1-tz_{L,k}}{z_{L,k}(t-z_{L,k})} - \frac{(1-tz_{L,k})^2}{(t-z_{L,k})^2}
\end{equation}
Since the $z_{L,k}$ are all on the unit circle and $0<t<1$, the asymptotics of this (for large $L$) are
\begin{equation}
D_{L,k}(t) = \frac{(1-tz_{L,k})(1-z_{L,k}^2)}{z_{L_k}(t-z_{L,k})}L + O(1)
\end{equation}
Now
\begin{equation}
m(z) = \frac{1-tz}{t-z}
\end{equation}
is a M\"obius transformation which maps the unit circle to itself. Hence for $z=e^{i\theta}$ on the unit circle,
\begin{equation}
\left|\frac{(1-tz)(1-z^2)}{z(t-z)}\right| = \left|\frac{1-z^2}{z}\right| = 2|\sin\theta|.
\end{equation}
Let us assume we choose all the $z_{L,k}$ to be in the upper half unit circle. Then if we can show that none of the roots are too close to $\pm 1$, i.e.\ there are no roots of the form $z=e^{i\theta}$ with $\theta$ close to 0 or $\pi$, then $|D_{L,k}|$ cannot be very small. We will show that if $z_{L,k} = e^{i\theta}$ with $0 < \theta < \pi$, then in fact
\begin{equation}\label{eqn:theta_range_roots}
\frac{\pi}{L+1} \leq \theta \leq \pi - \frac{\pi}{L+1},
\end{equation}
from which it follows that
\begin{equation}
|D_{L,k}(t)| \geq 2\pi + O(L^{-1}).
\end{equation}
(This bound is in fact tight -- if we order the roots for $k=3,\ldots,L+1$ anticlockwise from right to left, then at $k=3$ and $k=L+1$ we have $|D_{L,k}(t)| \to 2\pi$ as $L\to\infty$, for all $t>1$. We will make no attempt to prove this here, however.)

We wish to show that if $z=e^{i\theta}$ with $0<\theta<\frac{\pi}{L+1}$ or $\pi-\frac{\pi}{L+1} < \theta < \pi$, then $z$ cannot be a root of $A_{L,k}(t,z)=0$. There are a number of cases, which we will briefly consider in turn. 

\paragraph{I. Odd $k$, small $\theta$.} We have $z^{L+1}=m(z)$.  If $0<\theta < \frac{\pi}{L+1}$ then the LHS will be in the upper half of the unit circle. The RHS is a M\"obius transformation which maps the upper half of the unit circle to the lower half, so there is no root.

\paragraph{II. Odd $k$, even $L$, large $\theta$.} Write $\theta = \pi - \phi$. Then
\begin{equation}
z^{L+1} = e^{i\theta(L+1)} = e^{i\pi(L+1)}e^{-i\phi(L+1)} = -e^{-i\phi(L+1)}
\end{equation}
which is again in the upper half of the unit circle.

\paragraph{III. Odd $k$, odd $L$, large $\theta$.} This time $z^{L+1} = e^{-i\phi(L+1)}$ which is in the lower half of the unit circle. However, observe that as $\phi \nearrow \frac{\pi}{L+1}$ we have 
\begin{equation}
\arg(z^{L+1}) \searrow -\pi < \arg(m(z))
\end{equation}
(where we take arguments to be between $-\pi$ and $\pi$). Then since
\begin{equation}
\frac{d}{d\theta} \arg(m(e^{i\theta})) = \frac{t^2-1}{t^2+1-2t\cos\theta} < 1
\end{equation}
while
\begin{equation}
\frac{d}{d\theta} \arg(e^{i\theta(L+1)}) = L+1,
\end{equation}
we must have $\arg(z^{L+1}) < \arg(m(z))$, so there can be no root for $\pi-\frac{\pi}{L+1} < \theta < \pi$.

\paragraph{IV. Even $k$, odd $L$, large $\theta$.} With even $k$ we have $z^{L+1} = -m(z)$. The RHS is now a M\"obius transformation mapping the upper half of the unit circle to itself.  The rest of this case is analogous to case \textbf{II} above.

\paragraph{V. Even $k$, small $\theta$.} This uses the same argument as case \textbf{III} above -- one shows that $\arg(z^{L+1}) > \arg(-m(z))$.

\sloppy
\paragraph{VI. Even $k$, even $L$, large $\theta$.} Similar to cases \textbf{III} and \textbf{V} above, this time showing that $\arg(z^{L+1}) < \arg(-m(z))$
\fussy

\sloppy
\printbibliography
\fussy

\end{document}